%% file: paper.tex
\documentclass[]{eptcs}
\providecommand{\event}{CREST 2016}

\usepackage[utf8]{inputenc}  
\usepackage[T1]{fontenc}     
\usepackage[english]{babel}  
\usepackage{graphicx}        
\usepackage{microtype}       
\usepackage{cite}            
\usepackage{hyperref}        
\usepackage{comment}         
\usepackage{breakurl}        

\title{Fault Localization in Web Applications via Model Finding}

\author{%
Sylvain Hallé, Oussama Beroual\institute{%
Laboratoire d'informatique formelle\\
Université du Québec à Chicoutimi, Canada\\
}
}

\def\titlerunning{Fault Localization in Web Applications via Model Finding}
\def\authorrunning{S.\ Hallé \& O.\ Beroual}
\input includes.tex

\graphicspath{{fig/}}

\begin{document}

\maketitle
\begin{abstract}
\input abstract.tex
\end{abstract}

\section{Introduction} 

The use of the web has seen significant changes since the inception of HTML in the early 1990s. The web has now become a large and mature software ecosystem on par with the complexity found in traditional desktop applications. However, due to the somewhat complex relationship between HTML, CSS and Ja\-va\-Script, the layout of web applications tends to be harder to properly specify in contrast with traditional desktop applications. The same document can be shown in a variety of sizes, resolutions, browsers and even devices, making the presence of so-called layout ``bugs'' all the more prevalent. Such problems can range from relatively mundane quirks like overlapping or incorrectly aligned elements, to more serious issues compromising the functionality of the user interface.

In past work, we developed an automated tool for the detection of such bugs, evaluating expressions in a high-level declarative language based on first-order and linear temporal logic \cite{DBLP:conf/icst/HalleBGB15}. However, it was recognized early on that the basic evaluation of such properties, returning a simple true/false verdict, would not prove very helpful to a designer: web pages are composed of hundreds of elements with dozens of properties each, and over which many assertions are required to hold. What is more, sometimes the layout faults are too subtle to be visible to the naked eye (such as elements off by a single pixel). To provide real value to practitioners, a layout analysis tool should hence be able to pinpoint specific elements of the page that are responsible for some bug. Our work on layout bug detection has therefore turned into a form of \emph{fault localization}.

In Section \ref{sec:layout}, we first show examples of layout bugs in real-world web applications, and describe an early attempt at building an explanation for the violation of a specification in that context. This process builds what is called a ``witness''; its construction is based on a function applied recursively on the \emph{formula} that is falsified. A witness highlights a set of elements in the page which, in some way, are related to the violation of a property. 

This has proved insufficient in practice; therefore, in Section \ref{sec:repairs}, we describe ongoing work on a new formal grounding, this time based on the concept of ``repair''. Intuitively, a repair is a minimal set of transformations which, when applied to the original object, restores its satisfiability with respect to the specification. The advantage of this concept is that it is, at its highest level, independent from the nature of the object and the specification language used to declare its expected properties. It could hence be applied to a variety of other scenarios, besides web applications. Section \ref{sec:discussion} further discusses the concept of repairs, and provides a basic algorithm for computing them. 
Based on a few examples, we show how it presents the potential of improving the task of fault detection and correction, by automatically providing hints that correspond to intuition.



\section{Fault Localization in Web Applications}\label{sec:layout} 

A \emph{layout-based bug} is a defect in a web system that has visible effects on the content of the pages served to the user. This content can be anything observable by the client, including the structure of the page's elements and the dimensions and style attributes of these elements. In past work, we studied bugs in web applications that can be detected by analyzing the contents and layout of page elements inside a browser's window. Based on an empirical analysis of 35 real-world web sites and applications (such as Facebook, Dropbox, and Moodle), we provided a survey and classification of more than 90 instances of layout-based bugs.

\subsection{Examples of Layout Bugs}

Figure \ref{fig:layout-examples} shows two examples of frequent layout bugs we encountered. In the first screenshot, two elements of the page that should be disjoint are actually overlapping. In that particular instance, the problem is caused by the fact that elements are absolutely positioned in a page with respect to their dimensions when the text they contain is in English. When displaying the web site in another language (such as French), it may occur that the corresponding text is longer than the English version, causing two elements that were disjoint to suddenly overlap.

\begin{figure}
\centering
\subfloat[]{\frame{\includegraphics[width=2.5in]{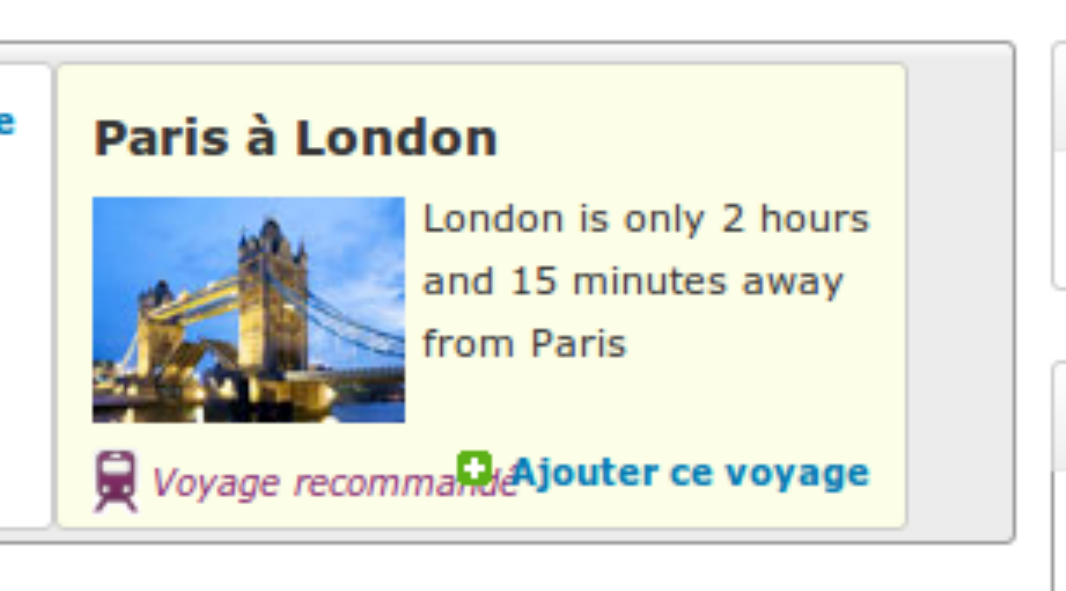}}}~~~
\subfloat[]{\frame{\includegraphics[width=2in]{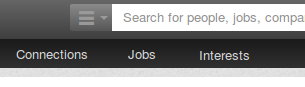}}}
\caption{Two examples of real-world layout bugs in web applications: (a) overlapping elements; (b) incorrectly aligned elements.}
\label{fig:layout-examples}
\end{figure}

Another mundane but frequent layout problem is the presence of elements which visibly should be aligned but are not. Figure \ref{fig:layout-examples}b shows an instance of this bug for the web platform LinkedIn. Sometimes, the misalignment is subtle, as an element may be off by a single pixel, as is the case here.

We then developed \appname{}, an automated testing tool that provides a declarative language to express desirable properties of a web application as a set of human-readable assertions on the page's HTML and CSS data \cite{DBLP:conf/icst/HalleBGB15}. Such properties can be verified on-the-fly as a user interacts with an application.

The core of the Cornipickle language is a high-level, English-like grammar that translates in the background into first-order linear temporal logic. For example, one can express that all elements of some list with ID ``menu'' should be vertically aligned by writing the following expression:

\begin{center}
\begin{minipage}{2in}
\begin{verbatim}
For each $x in $(#menu li) (
  For each $y in $(#menu li) (
    $x's left equals $y's left
)).
\end{verbatim}
\end{minipage}
\end{center}

\noindent
Here \verb+$(#menu li)+ is a called CSS \emph{selector}. A selector is a filter expression designating a set of elements in a web page, as defined in the CSS specification \cite{css-ref}. For example, the expression \verb+p > li.foo+ designates all elements with name \verb+li+ and class attribute \verb+foo+ immediately nested within an element of name \verb+p+. The syntax \verb+$x's left+ is used to refer to the CSS \verb+left+ property of an element, which designates the absolute $x$ coordinate of its top-left corner within the page. For this particular example, Figure \ref{fig:list-wrong-a} shows a simple page for which the property would be violated.

\begin{figure}
\centering
\subfloat[]{\includegraphics[height=0.75in]{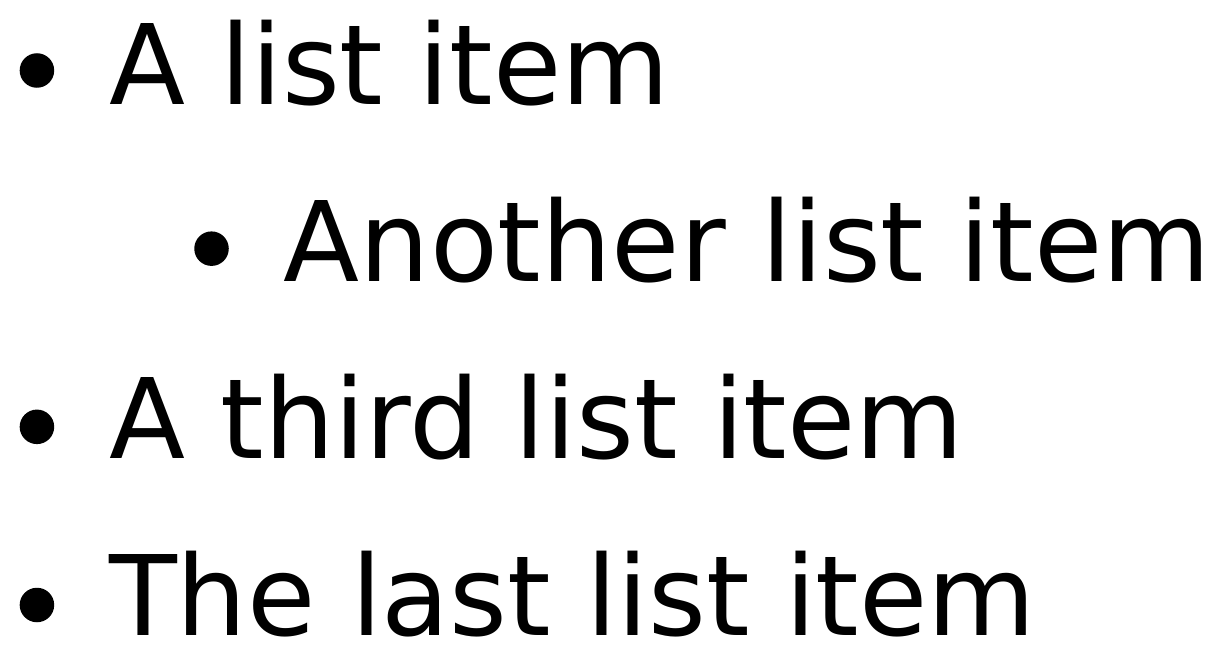}\label{fig:list-wrong-a}}~~~
\subfloat[]{\includegraphics[height=0.75in]{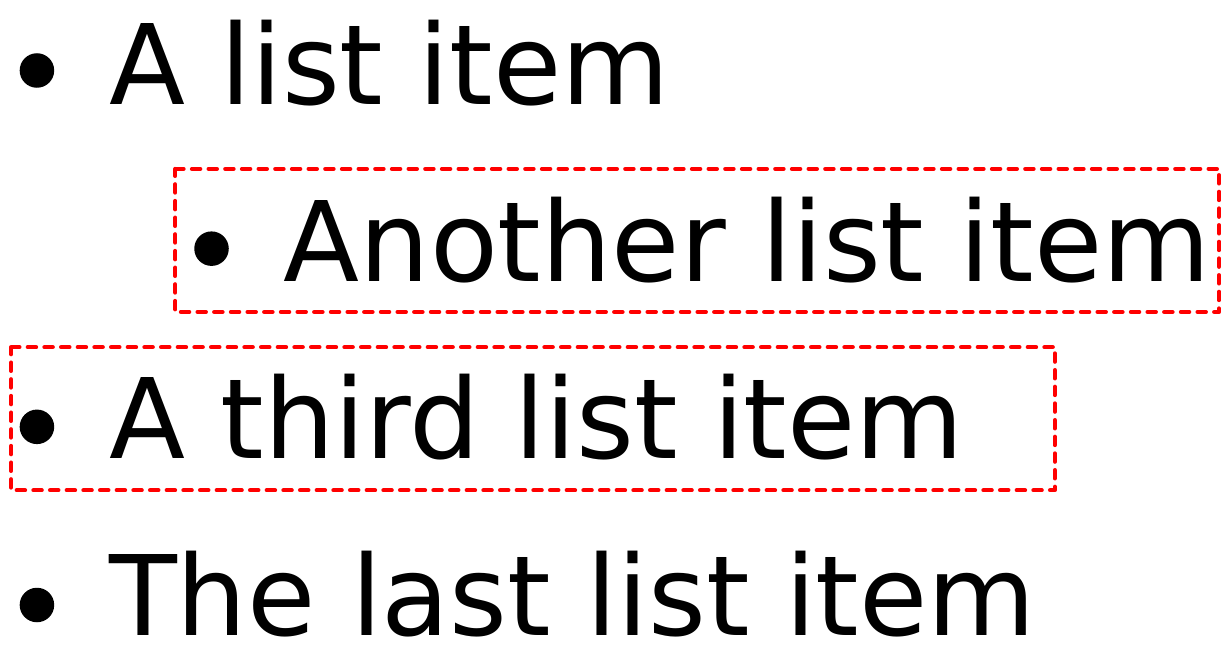}\label{fig:list-wrong-b}}
\caption{Example of a simple web layout fault: (a) one of the list items is incorrectly aligned with the others; (b) a witness produced by the \appname{} tool.}
\label{fig:list-wrong}
\end{figure}

\subsection{Witnesses}

\appname{} can easily evaluate properties of this kind on the contents of a page and report a true/false verdict. However, while in some cases the reason for the violation of a property is obvious, in many others a mere Boolean outcome bears little meaning for a web designer. After all, some errors can be subtle, or refer to a single element in a page that contains hundreds. Being simply told that ``something is wrong'' provides little added value when someone is required to hunt for the problem in such a complex page. 

To this end, \appname{} was fitted with a mechanism for attempting to circumscribe parts of a page that account for the discovered fault, in the form of what is called a \emph{witness}. A witness is  tree of DOM elements; let $W$ be the set of all witnesses. The set of \emph{verdicts} is defined as $V : \{\top,\bot,?\} \times W \times W$; a verdict is composed of a truth value and two witnesses: one corresponding to truth value $\top$, the other to truth value $\bot$.

The verdict conjunction is a function $\otimes : V \times \Nu \times V \rightarrow V$ defined as follows:
\[
\otimes(\langle b, w_\top, w_\bot\rangle, \nu, \langle b', w'_\top, w'_\bot\rangle) = \begin{cases}
\langle \bot, w_\top, w_\bot \cup \{(\nu, w'_\bot)\} \rangle & \mbox{if $b' = \bot$} \\
\langle ?, w_\top \cup \{(\nu, w'_\top)\}, w_\bot\rangle & \mbox{if $b \neq \bot$ and $b' = ?$} \\
\langle b, w_\top \cup \{(\nu, w'_\top)\}, w_\bot \rangle & \mbox{if $b \neq \bot$ and $b' = \top$} \\
\langle b, w_\top, w_\bot\rangle & \mbox{otherwise}
\end{cases}
\]
The notation $(\nu, w)$ designates the creation of a new witness whose root is the DOM node $\nu$, with witness $w$ as its child. The notation $w \cup w'$ designates the addition of $w'$ to the children of witness $w$. We will abuse notation and accept that the second argument of $\otimes$ be some ``empty'' element of $\Nu$ we will designate as $\nu_\emptyset$.

Verdict conjunction updates the contents of an existing verdict $v$, given another verdict $v'$ and some DOM element $\nu$. If $v'$ is false, it carries a witness of that falsehood, namely $w'_\bot$; this witness is attached as a child of a new tree whose root is $\nu$, and that tree is added to $v$'s witness of falsehood, $w_\bot$. Moreover, $v$'s truth value is set to $\bot$. In other words, $v'$'s explanation for being false is added to $v$'s explanation for being false. Otherwise, if neither $v$ nor $v'$ is false, then $v'$'s witness associated to $\top$ is added to $v$'s $\top$ witness, and its truth value is updated accordingly. In all other cases, $v$ is left unchanged. A dual definition can be built for verdict disjunction.

Using these operators, the formal semantics of the language can be lifted to a function $\omega : T^* \times \Phi \rightarrow V$, which, out of an expression $\varphi \in \Phi$ and a DOM tree $\nothing{t} \in T^*$, computes a verdict. The recursive semantics of that function is shown in Table \ref{tab:witness}. The details of that function are out of the scope of the present paper. However, we can see the result of applying $\omega$ on the DOM tree of Figure \ref{fig:list-wrong-a}. The function returns a tree containing pointers to two of the page's elements, highlighted in red in Figure \ref{fig:list-wrong-b}. (Actually, the function returns multiple sets, each of which contains the second list item and one of the remaining items.)

\begin{table}\small
\input witness.inc.tex
\caption{The recursive definition of the verdict computation function $\omega$}
\label{tab:witness}
\end{table}

Intuitively, such a result makes sense for a web designer; indeed, these two elements should be aligned, while they are not. However, this information can only be deduced through knowledge of the violated property; the witness simply points to these two elements, without providing information about ``what is wrong'' about them.


\section{A Generic Definition of Repairs}\label{sec:repairs} 

While the recursive counter-example generation present in the current version of \appname{} provides more information than a simple true/false verdict, in many cases it may still prove too vague to be useful.

In this section, we introduce the notion of \emph{repair}, which can be defined intuitively as a set of modifications required to some object to make it satisfy a property. The notion of repair can be seen as fault localization, expressed in reverse: stating how an object needs to be repaired indirectly points to aspects of its structure that are responsible for the fact that the property is not currently fulfilled.
We shall see that, contrarily to the concept of witness, which is heavily coupled with the specification language and domain objects used, repairs are defined at a level of abstraction that does not rely on properties of either.

\subsection{Definition}

Let $\Sigma$ be a set of \emph{structures}, and $T_\Sigma$ a set of endomorphisms on $\Sigma$; that is, each $\tau \in T_\Sigma$ is a function $\tau : \Sigma \rightarrow \Sigma$. Let $2^{T_\Sigma}$ designate the set of all subsets of $T_\Sigma$. A set of endomorphisms $T = \{\tau_1, \dots, \tau_n\} \in 2^{T_\Sigma}$ is said to be \emph{well defined} if any two elements $\tau_i$, $\tau_j$ are such that $\tau_i \circ \tau_j \equiv \tau_j \circ \tau_j$. Such a well defined set will be called a \emph{transformation}. When the context is clear, we shall abuse notation and consider $T$ as the (uniquely defined) endomorphism $\tau_1 \circ \dots \circ \tau_n$. Set inclusion induces a partial ordering over transformations.

Let $\Phi$ be a set of \emph{language expressions} equipped with a satisfaction relation $\models~:~\Sigma \times \Phi \rightarrow \{\top,\bot\}$. For an expression $\varphi \in \Phi$ and a structure $\sigma \in \Sigma$, we will write $\sigma \models \varphi$ if and only if $\models(\sigma,\varphi) = \top$. In such a case, we shall say that $\sigma$ ``satisfies'' $\varphi$, or alternately that $\sigma$ is a \emph{model} of $\varphi$.

Let $\sigma \in \Sigma$ be a structure such that $\sigma \not\models \varphi$ for some expression $\varphi \in \Phi$. A \emph{repair} is defined as a transformation $T \in 2^{T_\Sigma}$ such that $T(\sigma) \models \varphi$. A repair is said to be \emph{prime} if no subset $T' \subseteq T$ is such that $T'$ is also a repair. Intuitively, a prime repair is a set of ``changes'' to a structure $\sigma$ that make it satisfy $\varphi$, such that no ``smaller'' change also restores satisfiability. Since $\subseteq$ is a partial order, there may be multiple, mutually incomparable prime repairs.

Figure \ref{fig:lattice} illustrates this concept. The picture represents all transformations that can be applied to a structure, in the simple case where only four morphisms exist. The empty transformation is at the bottom, and each arrow in the graph represents the addition of one more morphism to an existing transformation. Red nodes indicate transformations that are not repairs, while yellow and green nodes indicate repairs. Of these, prime repairs are coloured in green; one can see that all antecedents of green nodes are red. The converse, however, is not true: not all descendants of a repair are repairs themselves.

\begin{figure}
\centering
\includegraphics[width=2in]{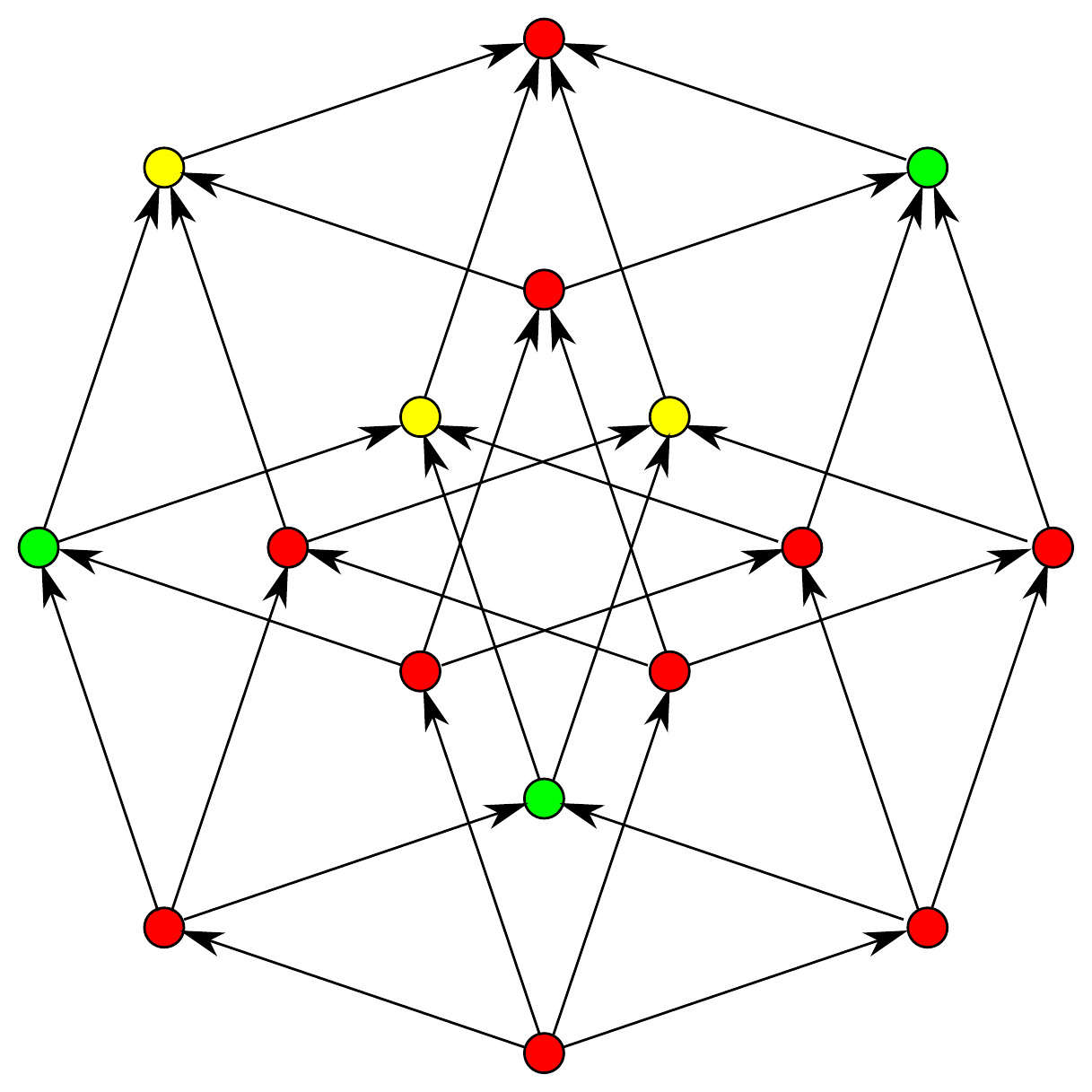}
\caption{Illustration of the concept of prime repair.}
\label{fig:lattice}
\end{figure}


\subsection{Examples}

This simple definition can then be applied to a variety of specification languages, as we shall illustrate through the examples that follow.

\subsubsection{Propositional Logic}

As a first example, let $\Phi$ be the set of propositional logic formul\ae{} with variables $X = \{x_1, \dots, x_n\}$ for some $n \geq 1$. Let $\Sigma$ be the set of functions $X \rightarrow \{\top,\bot\}$, which we shall call \emph{valuations}. The satisfaction relation $\models$ is defined as $\sigma \models \varphi = \top$ if and only if $\varphi$ evaluates to true when its variables are replaced by the corresponding truth value specified by $\sigma$, and $\bot$ otherwise.

Let $b \in \{\top,\bot\}$ and $i \in [1,n]$. We will note $\tau_{x_i \mapsto b}$ the endomorphism defined as:
\[
(\tau_{x_i \mapsto b}(\sigma))(x) = 
\begin{cases}
b & \mbox{if $x = x_i$} \\
\sigma(x) & \mbox{otherwise}
\end{cases}
\]
This morphism sets $x_i$ to $b$ and leaves the rest of the original valuation unchanged. The set of endomorphisms $T_\Sigma$ is then defined as:
\[
T_\Sigma = \bigcup_{i \in [1,n]} \bigcup_{b \in \{\top,\bot\}} {\tau_{x_i \mapsto b}}
\]

Two transformations $\tau_{x\mapsto b}$, $\tau'_{y\mapsto b'}$ commute if $x \neq y$. Hence a set of transformations $T \in 2^{T_\Sigma}$ is well defined if and only if every endomorphism it contains changes the value of a different variable.

\begin{example}
Let $X = \{a, b, c\}$, $\sigma$ be the valuation $\{a \mapsto \top, b \mapsto \bot, c \mapsto \bot\}$ and $\varphi$ the propositional formula $a \wedge b$. One can easily observe that $\sigma \not\models \varphi$. A repair is the transformation $T= \{\tau_{b \mapsto\top}\}$; that is, $T(\sigma) \models \varphi$	. This corresponds to the intuition that the explanation for the falsehood of $\varphi$ is that $b$ is false while it should be true. Note that although $T'=\{\tau_{b \mapsto\top}, \tau_{c \mapsto\top}\}$ would also make $\varphi$ true, it is not a \emph{prime} repair, since $T \subseteq T'$. This corresponds to the intuition that the truth value of $c$ is not relevant to the falsehood of $\varphi$.
\end{example}

\begin{example}\label{ex:prop2}
Let $\sigma$ be the valuation $\{a \mapsto \top,b \mapsto \bot,c \mapsto\bot\}$ and $\varphi$ the propositional formula $a \rightarrow b$. This time, two prime repairs exist: $T= \{\tau_{b \mapsto \top}\}$ and $T' = \{\tau_{a \mapsto\bot}\}$. It is possible to check that both fix the truth value of the original valuation. Informally, the first transformation accounts the falsehood of $\varphi$ on the fact that $a$ is true, while the other one rather explains it by the fact that $b$ is false ---which indeed corresponds to the intuition. Since both repairs are incomparable, none of these explanations is ``preferred''. We shall revisit this concept later.
\end{example}

\subsubsection{First-Order Logic}

The concept of repair can easily be lifted to the set $\Phi$ of first-order logic formul\ae{} on finite domains. Let $A$ be a set of elements; an $n$-ary predicate is defined as a function $p : A^n \rightarrow \{\top,\bot\}$; let $P^i$ be the set of predicates of arity $i$. A signature is a set of predicates $P = \{p_1, \dots, p_m\}$, respectively of arity $a_1, \dots, a_m$. For a given signature, the set of domain elements is defined as:
\[
\Sigma = P^{a_1} \times \dots \times P^{a_m}
\]

The satisfaction relation $\models$ is defined as $\models(d, \varphi) = \top$ if $\varphi$ evaluates to true when evaluating predicates as defined in $\sigma$, and $\bot$ otherwise.

In this context, an endomorphism will represent the change in the truth value for one input of one predicate. Let $p_k$ be a predicate of arity $i$, $(a_1, \dots, a_k) \in A^n$ be a $k$-tuple of elements of $A$, and $b \in \{\top,\bot\}$. The transformation $\tau_{p_k(a_1,\dots,a_k)\mapsto b}$ is defined as the predicate $p_k'$ such that:
\[
p_k'(x_1,\dots,x_k) = 
\begin{cases}
b & \mbox{if $x_1 = a_1$, \dots, $x_n = a_n$}\\
p_k(x_1,\dots,x_k) & \mbox{otherwise}
\end{cases}
\]

The set of transformations for $p_k$, noted $T_{p_k}$, is defined as:
\[
T_{p_k} \triangleq \bigcup_{(a_1,\dots,a_k)\in A^n} \left(\bigcup_{b \in \{\top,\bot\}} \{\tau_{p_k,(a_1,\dots,a_k),b}\}\right)
\]

\noindent The global set of transformations is then:
\[
T_\Sigma \triangleq \bigcup_{p \in P} T_{p}
\]

Similarly to first-order logic, one can check that two endomorphisms commute if they operate on a different predicate, or change the value of a different input on the same predicate.

\begin{example}
Let $A = \{0,1,2\}$, $\varphi$ be the first-order formula $\forall x : \exists y : x \neq y \wedge p(x,y)$, and the binary predicate $p$ defined as $\{(0,0), (0,1), (1,1)\}$. There are two prime repairs for restoring the truth of $\varphi$: $T_1=\{\tau_{p(2,0)\mapsto\top}\}$, $T_2=\{\tau_{p(2,1)\mapsto\top}\}$. This corresponds to the intuition that value 2 is missing at least one ``partner'' in $p$, and that either 0 or 1 could fit that purpose.
\end{example}

\begin{example}\label{ex:graph}
Let $A = [1,5]$ be a set of graph vertices, $p$ a binary predicate encoding the adjacency relationship of graph edges, and $q_1$,$q_2$,$q_3$ a set of unary predicates such that $q_i(x)$ holds if and only if vertex $x$ has colour $i$. Suppose predicates $p$ and $q$ are defined according to the graphical representation shown in Figure \ref{fig:graph-original}.

\begin{figure}
\begin{center}
\subfloat[Original graph]{\includegraphics[width=1.5in]{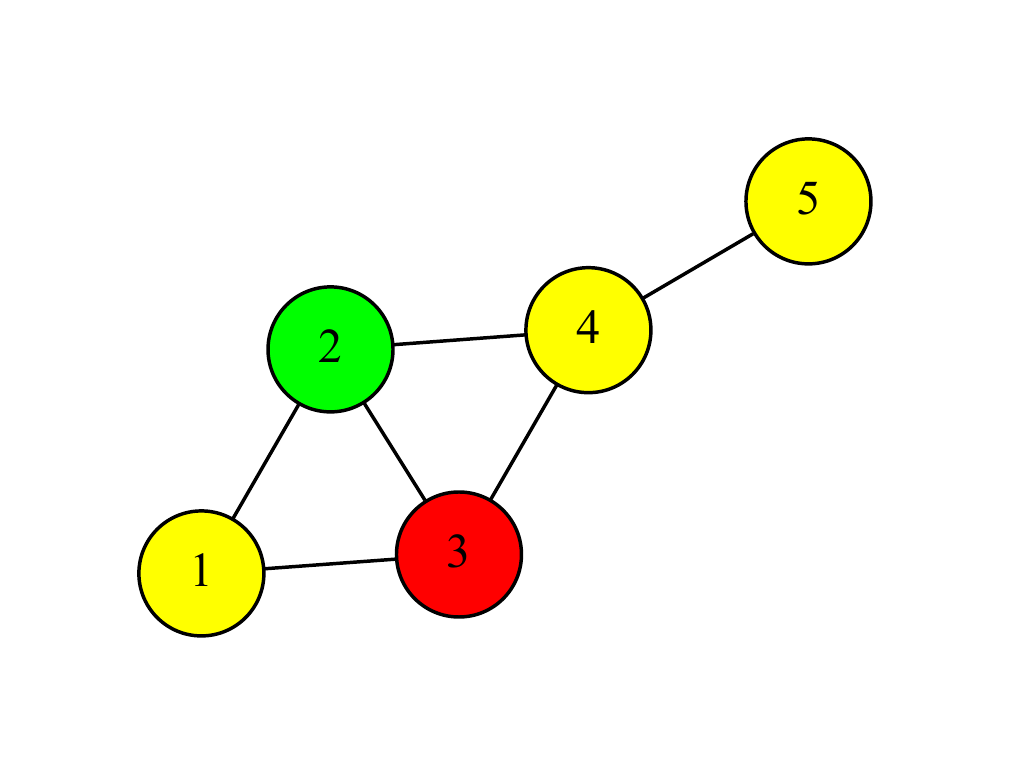}\label{fig:graph-original}}
\subfloat[After applying $T_1$]{\includegraphics[width=1.5in]{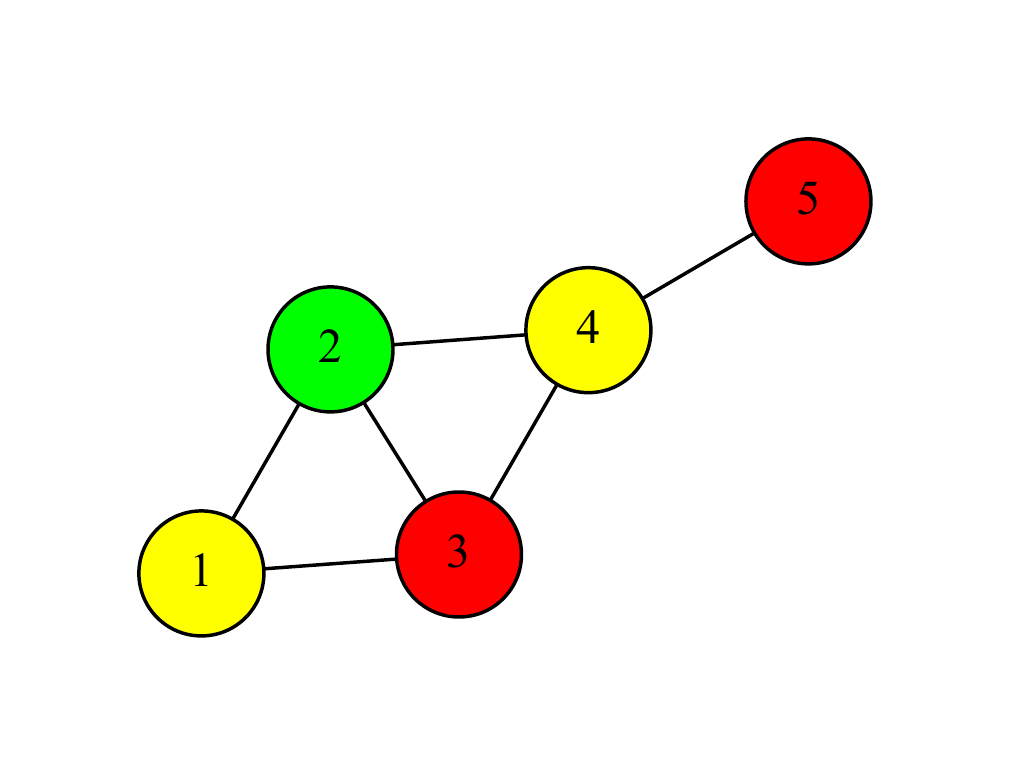}\label{fig:graph-t1}}
\subfloat[After applying $T_2$]{\includegraphics[width=1.5in]{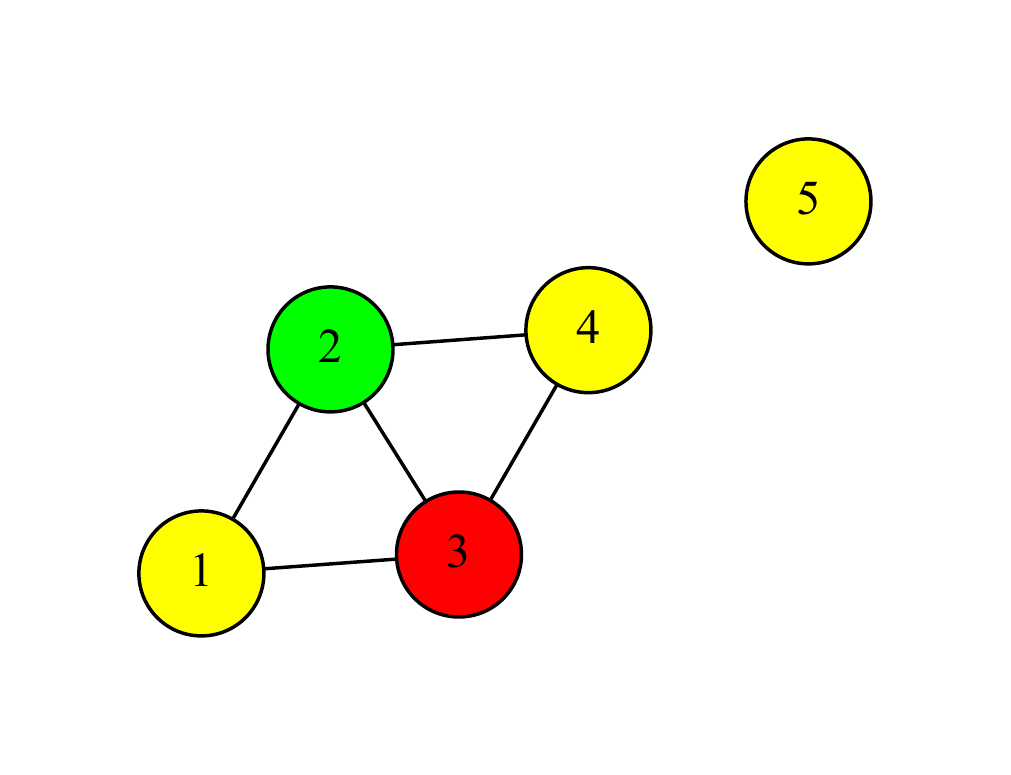}\label{fig:graph-t3}}
\\
\subfloat[After applying $T_3$]{\includegraphics[width=1.5in]{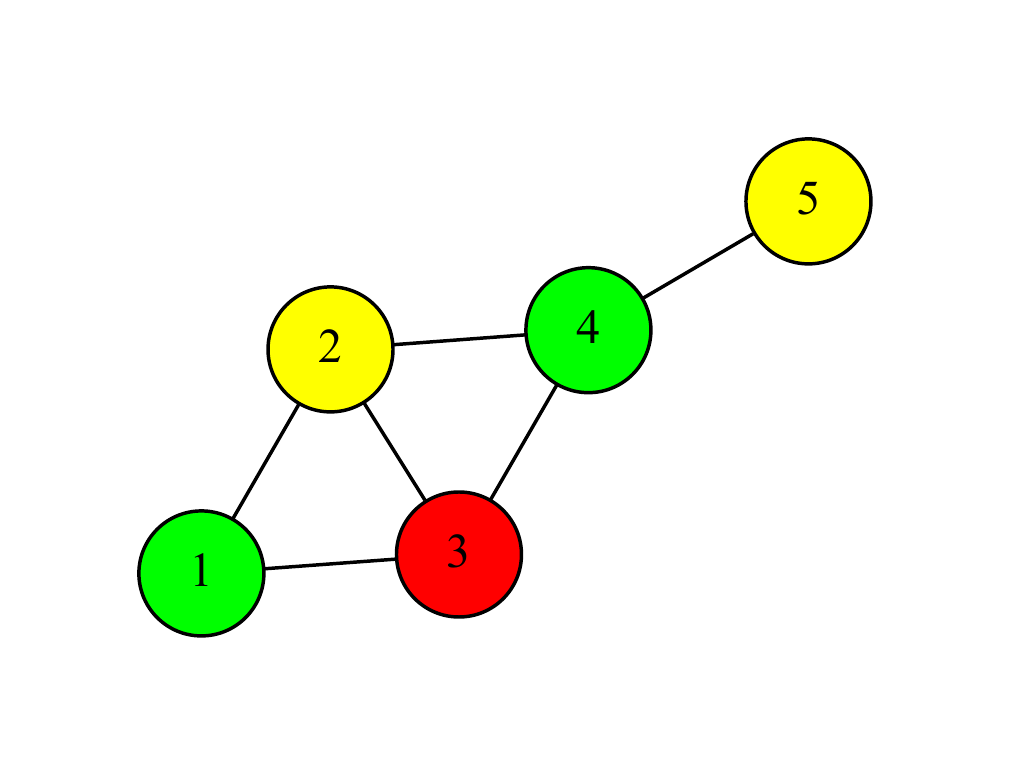}\label{fig:graph-t4}}
\subfloat[After applying $T_4$]{\includegraphics[width=1.5in]{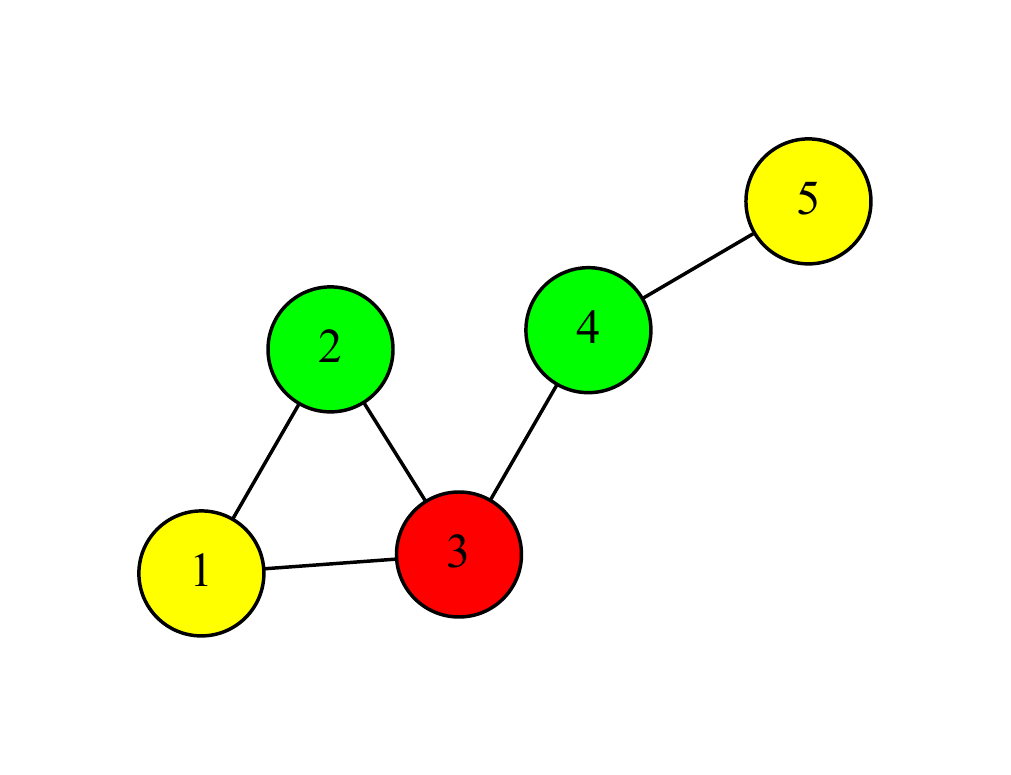}\label{fig:graph-t5}}
\caption{A few possible repairs for a faulty graph colouring}
\end{center}
\end{figure}

A solution to the graph colouring problem can be represented by three first-order expressions:

\begin{displaymath}
\varphi_1 \triangleq \forall x : (q_1(x) \wedge \neg q_2(x) \wedge \neg q_3(x)) 
\vee (\neg q_1(x) \wedge q_2(x) \wedge \neg q_3(x)) \vee (\neg q_1(x) \wedge \neg q_2(x) \wedge q_3(x))
\end{displaymath}
\begin{displaymath}
\varphi_2  \triangleq  \forall x : \forall y : p(x,y) \rightarrow p(y,x)
\end{displaymath}
\begin{displaymath}
\varphi_3  \triangleq  \forall x : \forall y : p(x,y) \rightarrow
((q_1(x) \rightarrow \neg q_1(y)) \wedge (q_2(x) \rightarrow \neg q_2(y)) \wedge (q_3(x) \rightarrow \neg q_3(y)))
\end{displaymath}

The first stipulates that every vertex has exactly one colour; the second indicates that the adjacency relation is symmetric, and the final expression stipulates that no adjacent vertices can have the same colour. One can see that the original graph does not satisfy $\varphi_1 \wedge \varphi_2 \wedge \varphi_3$. There exist multiple prime repairs, a few of which are shown here:
\begin{eqnarray*}
T_1 & = & \{\tau_{q_1(5)\mapsto\bot},\tau_{q_2(5)\mapsto\top}\} \\
T_2 & = & \{\tau_{p(4,5)\mapsto\bot},\tau_{p(5,4)\mapsto\bot}\} \\
T_3 & = &  \{\tau_{q_1(1)\mapsto\bot},\tau_{q_3(1)\mapsto\top},\tau_{q_1(4)\mapsto\bot},\tau_{q_3(4)\mapsto\top}\} \\
T_4 & = & \{\tau_{p(2,4)\mapsto\bot},\tau_{p(4,2)\mapsto\bot},\tau_{q_1(4)\mapsto\bot},\tau_{q_3(4)\mapsto\top}\}
\end{eqnarray*}

Repair $T_1$ fixes the graph by changing the colour of vertex 5 to red. Note that this necessitates not only setting $q_2(5)$ to $\top$, but also $q_1(5)$ to $\bot$; otherwise the resulting structure would violate $\varphi_1$. Another repair (not shown), changes vertex 5 to green. Repair $T_3$ rather alters the adjacency relation and cuts vertex 5 from the rest of the graph, so that the colour conflict is resolved.

These correspond to the ``intuitive'' ways of fixing the graph colouring. However, there exist multiple other prime repairs that fulfill the definition. For example, transformation $T_4$ exchanges the colours of vertices 1, 2 and 4. Note that this is indeed a prime repair, in that no subset of these endomorphisms restore satisfiability of the original formula. In the same way, $T_5$ cuts the edge between vertices 2 and 4, and turns 4 to green. In total, there are 17 distinct prime repairs in this particular example.
\end{example}

Again, it should be noted that without additional context, none of these repairs is a more likely explanation of the falsehood of $\varphi_1 \wedge \varphi_2 \wedge \varphi_3$ on the original graph.

\subsubsection{Extended First-Order Logic}

The previous example shows the need to extend the semantics of first-order logic to arbitrary functions instead of strictly Boolean predicates. This can easily be done as follows. Let $A_1, \dots, A_n$ and $B$ be finite sets. We will denote by $F^{A_1,\dots A_n\rightarrow B}$ the set of all functions $(\prod_i A_i) \rightarrow B$. 
A signature is a tuple of the form:
\[
\langle (A_{1,1}, \dots A_{1,n_1}) \rightarrow B_1, \dots, (A_{m,1}, \dots A_{m,n_m}) \rightarrow B_m\rangle
\]
such that $f_i$ is a function of arity $n_i$ with domain $A_{1,1}, \dots A_{1,n_i}$ and image $B_i$. Predicate logic is the special case where $B_1 = \dots = B_{n_m} = \{\top,\bot\}$, in which case the image can be omitted, and where the $A_{i,j}$ are all the same, so that only the arity needs to be known. If $f$ is a function $A \rightarrow B$ and $x$ designates an element of $A$, we shall write $x.f$ to denote $f(x)$, thus allowing some form of ``object'' notation for functions.

In this setting, first-order quantifiers need to precise over which of the $A_{i,j}$ they apply, so that expressions become of the form $\forall x \in A_{i,j} : \varphi$ and $\exists \in A_{i,j} : \varphi$. Ground terms can now compare values of two function terms, using any appropriate binary operator. 
Endomorphisms are still defined in the same way as for classical first-order logic, with the provision that they refer to appropriate values in the domain and image of the function subject to the change.

It should be noted that this extended formalism does not add any expressiveness to first-order logic if all sets are kept finite. It shall, however, simplify the expression of many properties.

\begin{example}\label{ex:web}
Equipped with this modified formalism, we are ready to consider repairs in web layout properties. Let $E$ be a set of page elements, $P$ be a set of pixel values, and $C$ be a set of CSS colours. Over these three sets, let us define the functions $E \rightarrow P$ called \verb+left+, \verb+right+, \verb+top+, and \verb+bottom+, corresponding to the $x$ and $y$ coordinates of the top-left and bottom-right corner of an element, respectively.
%
Additionally, we define a set $S$ of CSS \emph{selectors}; the evaluation of a CSS selector over a document can be formalized as a function $\$: S \rightarrow 2^E$ which, for a given filter expression, returns the subset of $E$ matching the selector. 

Endomorphisms can be defined for each of these functions, and shall be written using the notation introduced earlier. For example, $\tau_{\mbox{\scriptsize width}(e) \mapsto k}$ corresponds to the endomorphism setting the value of function \texttt{width} for element $e \in E$ to $k$, and leaving everything else as is.

One can then express the property that all items within a list with ID ``menu'' should be left-aligned as the following first-order expression:
\[
\forall x \in \$(\mbox{\tt \#menu li}) : \forall y \in \$(\mbox{\tt \#menu li}) : x.\mbox{\tt left} = y.\mbox{\tt left}
\]

\noindent
Note that this expression corresponds directly to the first-order translation of the \appname{} expression shown in Section \ref{sec:layout}.

\begin{figure}
\centering
\subfloat[]{\includegraphics[height=0.75in]{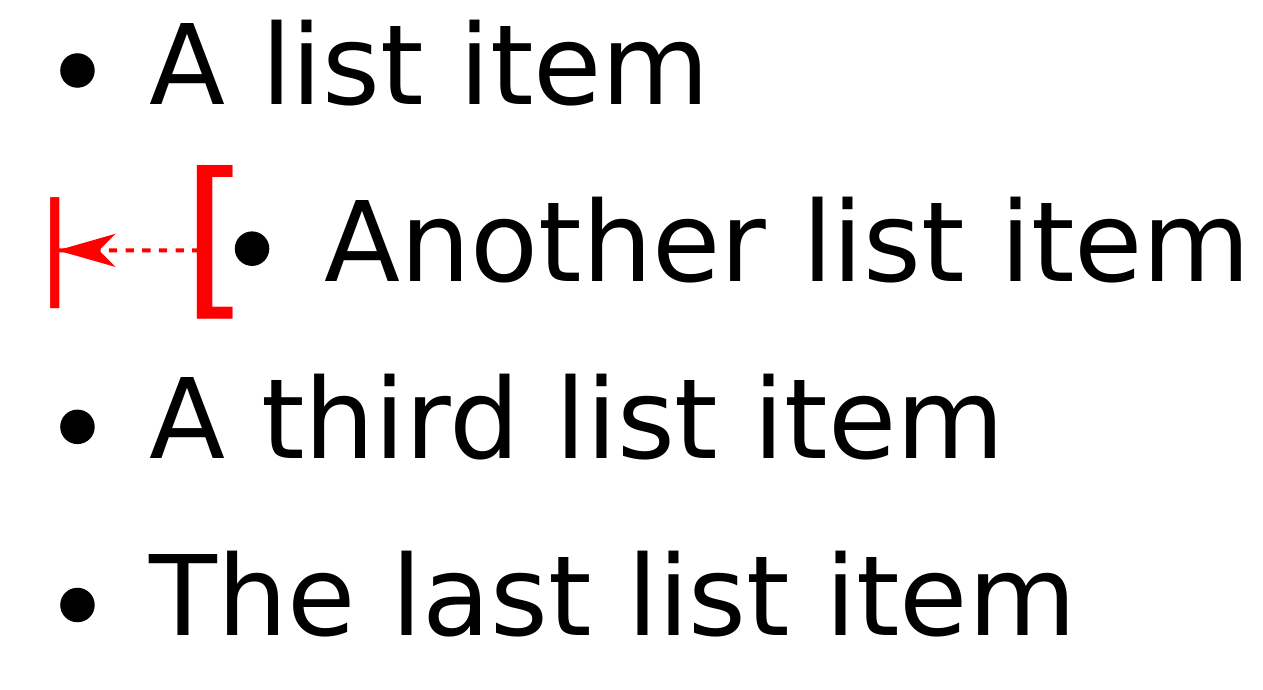}\label{fig:list-corrections-a}}~~
\subfloat[]{\includegraphics[height=0.75in]{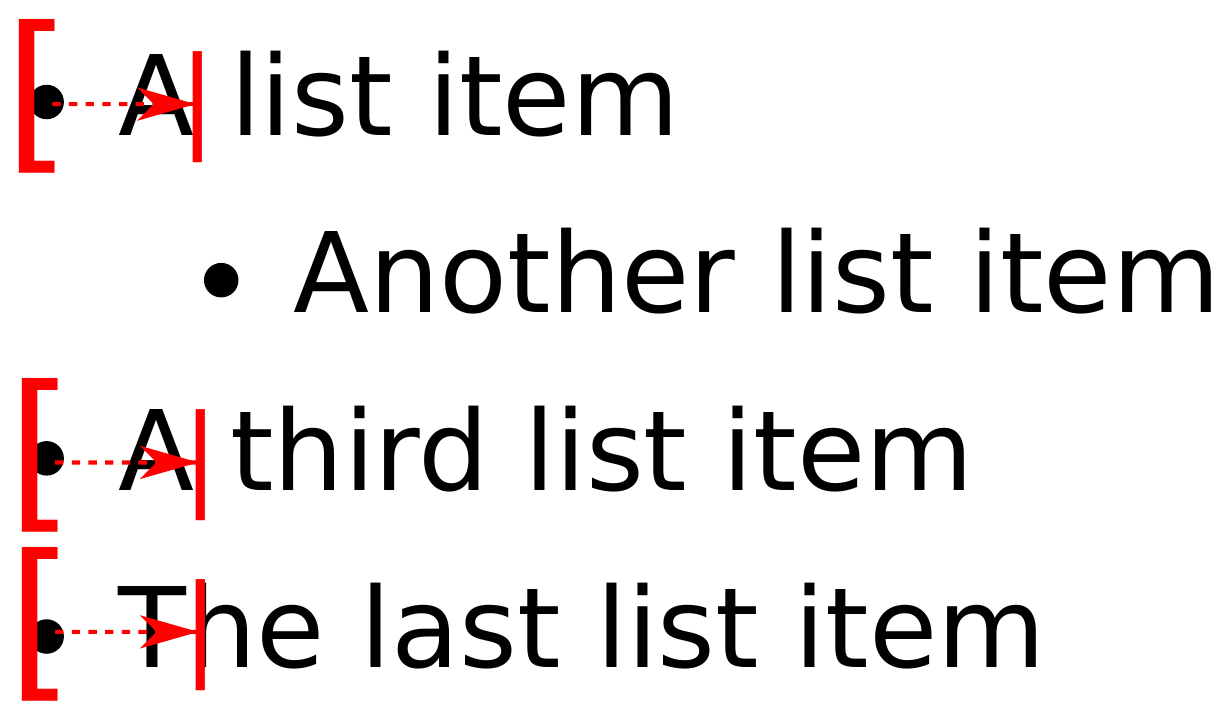}\label{fig:list-corrections-b}}~~
\subfloat[]{\includegraphics[height=0.75in]{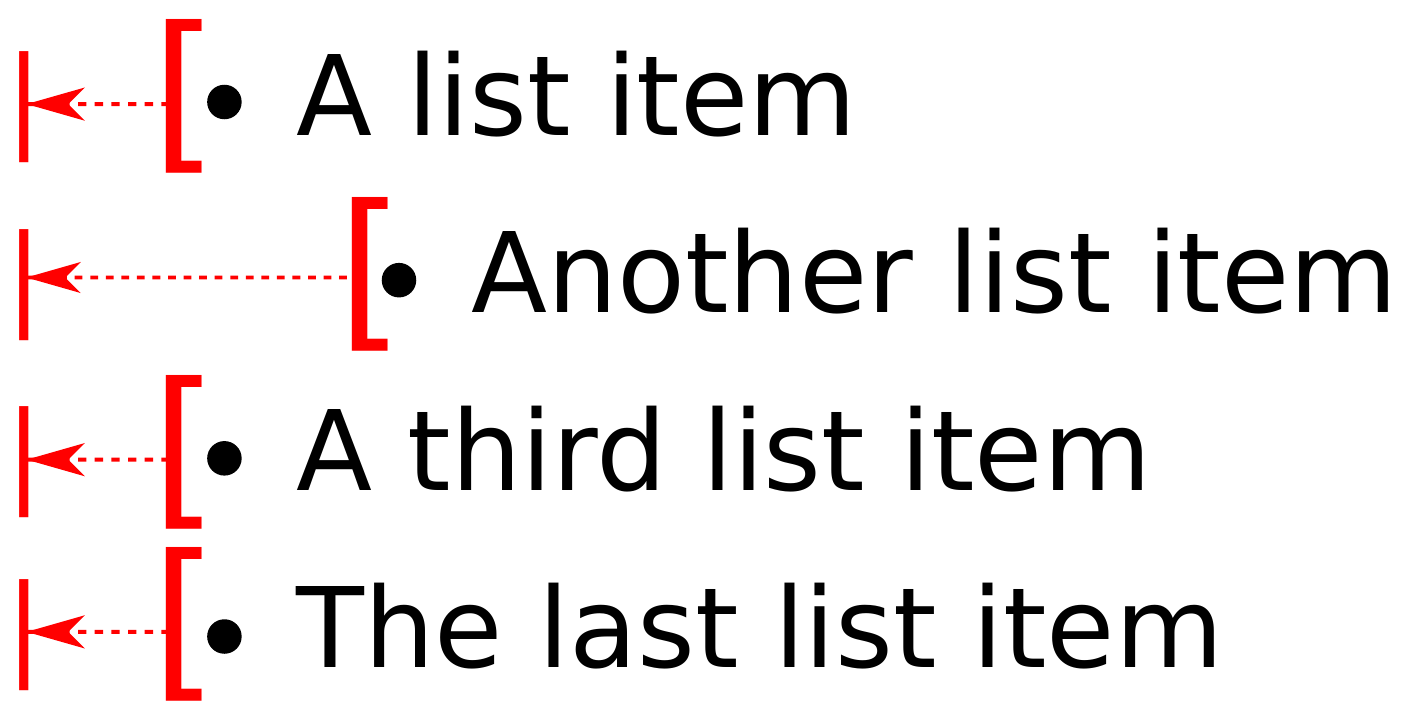}\label{fig:list-corrections-c}}
\caption{Three repairs for the web example}
\label{fig:list-corrections}
\end{figure}

Finding the prime repairs for that expression and the page fragment shown in Figure \ref{fig:list-wrong-a} produces a number of solutions, three of which are shown in Figure \ref{fig:list-corrections}. The first two are fairly intuitive. Figure \ref{fig:list-corrections-a} fixes the page by moving the lone misaligned list item in line with the others, while Figure \ref{fig:list-corrections-b} does the opposite, and aligns the three leftmost list items to the second one. Figure \ref{fig:list-corrections-c} gives an example of one of the many remaining solutions; in this case, all list items are moved to a new, common $x$ position, which turns out to be a coordinate that no element had in the original page.
\end{example}

This last example provides a graphical illustration of the difference between the original concept of witness, and that of repair. While a witness in this case highlights a randomly chosen pair of misaligned elements (as was shown in Figure \ref{fig:list-wrong-b}), a repair picks specific elements and, in addition, describes what should be done with them to fix the violation of the property. This is arguably more telling to a user, and constitutes in our view one of the key advantages of this technique.


%
%


\section{Computing Repairs}\label{sec:discussion} 

The basic concept of repair introduced in the previous section lends itself to a few discussion points. In particular, the number of possible prime repairs is potentially high, and the task of generating these repairs can therefore prove computationally intensive.

\subsection{Basic Algorithm and Complexity}

Algorithm \ref{alg:algo} shows an algorithm for iterating over all possible repairs of a structure. The algorithm simply enumerates all possible transformations $T \in 2^{T_\Sigma}$. It first checks whether $T$ is well defined (i.e.\, that any pair of endomorphisms commutes), and whether any previously generated repair (stored in set $T_S$) is a subset of the current one. It finally checks whether applying that transformation fixes the original structure. It skips to the next candidate transformation, should any of these three situations occur. Otherwise, the algorithm adds this transformation to its set, and returns it as its next element.

\begin{algorithm}
\caption{Generic algorithm for iterating over prime repairs}
\label{alg:algo}
\begin{algorithmic}[0]
\Procedure{ComputeRepairs}{$\varphi,\sigma,2^{T_\Sigma}$}
\State $T_S = \emptyset$
\ForAll{$T \in 2^{T_\Sigma}$}\Comment{Enumerated by increasing cardinality}
  \If{$\neg$\textsc{WellDefined}$(T)$}
  \State \textbf{skip}
  \EndIf
  \If{\textsc{Subsumed}$(T,T_S)$}
  \State \textbf{skip}
  \EndIf
  \If{$T(\sigma') \not\models \varphi$}
  \State \textbf{skip}
  \EndIf
  \State $T_S \gets T_S \cup \{T\}$
  \State \textbf{yield} $T$
\EndFor
\EndProcedure
\end{algorithmic}
\end{algorithm}

\begin{thm}
Algorithm \ref{alg:algo} is sound and complete.
\end{thm}
\begin{proof}
Let $T$ be a transformation output by the algorithm. By construction, $T$ is a repair, since is well defined and fixes the truth value of $\sigma$ on $\varphi$. Moreover, at the moment $T$ is output, it is such that none of the elements of $T_S$ are a subset of $T$. Since $T_S$ contains all repairs of cardinality smaller than $T$, and that, by construction, all transformations of similar cardinality cannot be subsets of each other, it follows that $T$ is not subsumed by any existing repair, and is hence prime. This proves the soundness of the algorithm.

The fact that all such prime repairs are eventually enumerated is guaranteed by the fact that all subsets of $T_\Sigma$ are generated at some point, thereby proving completeness.
\end{proof}

This algorithm has been implemented in Java and is publicly available\footnote{\url{https://bitbucket.org/sylvainhalle/fault-finder}}. Owing to its simplicity and its genericity, the implementation of expressions, structures and repair iteration amounts to a mere 325 lines of code. The enumeration of repairs is exposed to the user in the form of a classical Java \texttt{Iterator} class, which can be used through the traditional \texttt{hasNext()} and \texttt{next()} methods to pass through the entire set of prime repairs, in increasing order of cardinality. Domain-specific classes defining propositional and first-order logic constructs are made of roughly 500 additional lines of code.

It is easy to see that the running time of this algorithm is exponential in the size of $T_\Sigma$, which itself may be exponential in some other factor. 
In first-order logic, if $a_1, \dots, a_n$ is the respective arity of each predicate in the signature, the number of endomorphisms is $\sum_i 2|A|^{a_i}$ for a given domain $A$.

Despite this, it is possible to show that this algorithm is limited by a theoretical lower bound. A set of endomorphisms $T_\Sigma$ is said to be \emph{complete} if for every $\sigma, \sigma' \in \Sigma$, there exists a well-defined transformation $T \subseteq T_\Sigma$ such that $T(\sigma) = \sigma'$. 

\begin{thm}
Given a set of structures $\Sigma$, a set of language expressions $\Phi$ and a complete set of transformations $T_\Sigma$, the problem of computing prime repairs is at least as hard as the satisfiability problem for $\Phi$.
\end{thm}
\begin{proof}
Let $\varphi \in \Phi$ be some language expression. If $\varphi$ is satisfiable, then there exists some structure $\sigma \in \Sigma$ such that $\sigma \models \varphi$. Take an arbitrary structure $\sigma' \in \Sigma$. Since $T_\Sigma$ is complete, there exists at least one transformation $T \subseteq T_\Sigma$ such that $T(\sigma') = \sigma$. Take the smallest such set; by definition, it is a prime repair and will eventually be enumerated by Algorithm \ref{alg:algo}. Since the algorithm is sound and complete, on the contrary, no repair will be found if $\varphi$ is not satisfiable.
\end{proof}

\subsection{Reducing Number of Candidate Solutions}

These basic complexity results warrant a discussion about the reduction in the number of potential repairs that need to be explored.

\subsubsection{Removing Endomorphisms}

The number of potential transformations can first be reduced by removing endomorphisms that are known to be impossible, based on the context. For example, suppose that the propositional symbols $a$ and $b$ in Example \ref{ex:prop2} correspond to the assertions ``the client pays for an item'' and ``the client is shipped the item'', respectively. One could assume that a valuation where $a$ is true cannot be modified by making it false; this would correspond to the fact that an action done by some actor cannot be undone. In such a context, only endomorphisms setting false variables to true would be considered.

In the case of graphs, as in Example \ref{ex:graph}, one could impose restrictions on what changes are allowed to it; for example, one could say that existing edges must remain unchanged, or that only specific vertices may be coloured differently. This, again, has for effect of preferring some transformations over others, and globally reduces the number of available repairs.

\subsubsection{Transformations in Bulk}

The granularity of available endomorphisms can also be changed. In the case of the graph colouring example, it is obvious that no repair will ever consist of a single endomorphism $\{\tau_{q_i(x)\mapsto\top}\}$. The reason is that expression $\varphi_1$ requires that each vertex be of exactly one colour; assigning some $q_i$ to $\top$ for a vertex entails that the remaining $q_j$ for $j\neq i$ be set to $\bot$. One can hence define a new set of transformations appropriate for the context, representing \emph{colour changes}:
\[
T_C = \bigcup_{x \in A} \bigcup_{\substack{i\in[1,3]\\ j \neq i\\ k\neq j\neq i}} \{\{\tau_{q_i(x)\mapsto\top}, \tau_{q_j(x)\mapsto\bot}, \tau_{q_k(x)\mapsto\bot}\}\}
\]

Similarly, as the adjacency relation is symmetric, setting $p(x,y)$ to $\top$ (resp.\ $\bot$) cannot be done without also setting $p(y,x)$ to $\top$ (resp.\ $\bot$). Instead of considering individual changes to single inputs of $p$, one can define a set of \emph{edge changes}:
\[
T_E = \bigcup_{x \in A} \bigcup_{y \in A} \bigcup_{b \in \{\top,\bot\}} \{\{\tau_{p(x,y)\mapsto b}, \tau_{p(y,x)\mapsto b}\}\}
\]
One could then use $T_C \cup T_E$ as the set of transformations, instead of $T_\Sigma$. While this makes no change in theory on the actual solutions, the fact that $T_C \cup T_E$ is smaller in size than $T_\Sigma$ has a positive effect on the performance of an enumeration algorithm in practice.

The same can be said of the endomorphisms of Example \ref{ex:web}. Rather than consider all individual changes of $(x,y)$ coordinates of all four corners of every element, one could define subsets corresponding to more intuitive modifications; for example, the set of \emph{horizontal displacements} could be defined as:
\[
T_H = \bigcup_{e \in E} \bigcup_{p \in P} \{\{\tau_{\mbox{\scriptsize left}}(e) \mapsto p, \tau_{\mbox{\scriptsize right}}(e) \mapsto (\tau_{\mbox{\scriptsize right}}(e) - p)\}\}
\]

One can then restrict the search for repairs to those that are made only of (horizontal or vertical) displacements, or (horizontal or vertical) \emph{resizings} of elements, etc.


\section{Related Work} 

The concept of repair calculation is under construction, and its ties to related work still need to be established. As we have seen in the previous section, finding repairs relates to the concept of satisfiability solving (SAT), and more precisely the  problem of \emph{incremental} SAT \cite{DBLP:conf/sat/NadelR12}. Traditional SAT solvers are required to find a single model of an expression is such model exists. In incremental SAT, a solver finds a first model of an expression, but can also be repeatedly asked to provide additional models. When a set of transformations is complete, iterating over models amounts to iterating over repairs.

The use of model finding has also been studied in the field of network configuration management. An early version of the concept of witness was suggested by one of the authors \cite{DBLP:conf/ifip6/HalleWVC06} in the context of self-configuration. Couch \textit{et al.} proposed the notion of \emph{convergent} actions for repairing the configuration of a network \cite{DBLP:conf/dsom/CouchS03} (similar to our notion of well-definedness), while Narain \cite{DBLP:journals/jnsm/NarainLMK08} suggested six uses for a procedure $P(\varphi,x)$ that finds values for $x$ that satisfy some configuration constraint $\varphi$. Two of them deserve special attention:

\begin{itemize}
\item Configuration Error Fixing: if a configuration does not fulfill some specification $\varphi$, compute $P(\varphi, x )$ and take the solution $x$ closest to the current configuration.

\item Requirement Strengthening: to reconfigure a network with a set of additional constraints $\psi$, compute $P(\varphi \wedge \psi,x)$ (that is, the constraints of $\varphi$ and $\psi$ together) and take the result $x$.
\end{itemize}

In both cases, Narain suggested the use of a satisfiability solver for generating the target configuration; however, a SAT solver in general does not give control over what model is returned. The solution ``closest'' to the current configuration is better captured by the concept of prime repair, which corresponds exactly to that concept.


\section{Conclusion} 

The proof-of-concept prototype of \appname{} has shown promising results in its ability to easily express conditions for layout-based bugs in web applications, and efficiently detecting them in sample pages from more than 35 real-world applications. However, its ability to return a useful explanation for the violation of a property on a given web document is limited. This paper has introduced a definition of the concept of \emph{repair}, whose calculation provides more precise information about the changes required to a structure in order to satisfy a given specification.

The study of repairs and their computation is part of ongoing work, and many problems are still open. For example, an efficient computation of repairs relies on the  deletion of as many candidate transformations as possible; therefore, techniques to \emph{easily} identify endomorphisms that can never be part of a solution could be sought after. Similarly, we are planning to study techniques that could generate the set of repairs directly from the specification and the faulty structure, rather than using the crude generate-and-test algorithm presented in this paper.
In spite of this, early results on a number of examples show that using repairs as a form of fault localization is promising, and in particular provides results that correspond to intuition in many cases.

\input postamble-eptcs.inc.tex


\end{document}

%% file: includes.tex
\usepackage{amsmath,amsfonts,amssymb}
\usepackage{subfig}
\usepackage{algorithm} 
\usepackage{algpseudocode}
\usepackage{mathtools} 

\newsavebox{\verbbox}

\newcommand{\Nu}{N}

\usepackage[x11names]{xcolor}
\usepackage{todonotes}

\newtheorem{thm}{Theorem}

\usepackage{array}
\newcolumntype{x}[1]{%
>{\raggedleft\hspace{0pt}}p{#1}}%

\newcommand{\appname}{Cornipickle}

\newcommand{\nothing}[1]{#1}

\newenvironment{example}{}{}
\newenvironment{proof}{}{}

%% file: abstract.tex
We describe a generic technique for fault localization independent from the nature of the object or the specification language used to declare its expected properties. This technique is based on the concept of ``repair'', a minimal set of transformations which, when applied to the original object, restores its satisfiability with respect to the specification. We show how this technique can be applied with various specification languages, including propositional and finite first-order logic. In particular, we focus on its use in the detection of layout faults in web applications. 

%% file: witness.inc.tex
\begin{tabular}{x{2in}cp{3in}}
  $\omega(\nothing{t}, \nu\mbox{\texttt{'s }}a\mbox{\texttt{ equals }} \nu'\mbox{\texttt{'s }}a' ) $ & = & 
  \begin{minipage}{1.8in}
    \noindent
    \begin{displaymath}
    \begin{cases}
      \langle \top, \{\nu, \nu'\}, \emptyset\rangle & \mbox{if $\nu(a) = \nu'(a')$} \\
      \langle \top, \emptyset, \{\nu, \nu'\}\rangle & \mbox{otherwise}
      \end{cases}
      \end{displaymath}
\end{minipage}
\\
%
%
$\omega(\nothing{t}, \nu\mbox{\texttt{'s }}a\mbox{\texttt{ equals }} v ) $ & = & 
\begin{minipage}{1.5in}
  \noindent
  \begin{displaymath}
    \begin{cases}
    \langle \top, \{\nu\}, \emptyset\rangle & \mbox{if $\nu(a) = v$} \\
    \langle \bot, \emptyset, \{\nu\}\rangle & \mbox{otherwise}
  \end{cases}
  \end{displaymath}
\end{minipage} \\
$\omega(\nothing{t}, \mbox{\texttt{Not}}\, \varphi ) $ & = &  $\ominus(\omega(\nothing{t}, \varphi), \nu_\emptyset) $ \\
$\omega(\nothing{t}, \varphi \mbox{\texttt{ And }}\, \psi ) $ & = &  $\otimes(\otimes(\langle\top, \emptyset, \emptyset{}\rangle, \nu_\emptyset, \omega(\nothing{t}, \varphi)), \nu_\emptyset, \omega(\nothing{t}, \psi)) $ \\
$\omega(\nothing{t}, \varphi \mbox{\texttt{ Or }}\, \psi ) $ & = &  $\oplus(\oplus(\langle\bot, \emptyset, \emptyset{}\rangle, \nu_\emptyset, \omega(\nothing{t}, \varphi)), \nu_\emptyset, \omega(\nothing{t}, \psi)) $ \\
$\omega(\nothing{t}, \mbox{\texttt{If }} \varphi \mbox{\texttt{ Then }} \psi ) $ & = &  $\oplus(\oplus(\langle\bot, \emptyset, \emptyset{}\rangle, \nu_\emptyset, \ominus(\omega(\nothing{t}, \varphi), \nu_\emptyset)), \nu_\emptyset, \omega(\nothing{t}, \psi)) $ \\
$\omega(\nothing{t}, \mbox{\texttt{There exists }} \xi \mbox{\texttt{ in}}$ \newline 
$\mbox{\texttt{\$(}}c\mbox{\texttt{) such that }} \varphi ) $ & = &  $\bigoplus_{\nu \in \chi(\nothing{t}_0, c)}^{\langle\bot, \emptyset, \emptyset{}\rangle} \omega(\nothing{t}, \varphi[\xi/\nu])$ \\
$\omega(\nothing{t}, \mbox{\texttt{For each }} \xi \mbox{\texttt{ in \$(}} c\mbox{\texttt{) }}\varphi ) $ & = &  $\bigotimes_{\nu \in \chi(\nothing{t}_0, c)}^{\langle\top, \emptyset, \emptyset{}\rangle} \omega(\nothing{t}, \varphi[\xi/\nu])$ \\
%
%
%
%
\end{tabular}

%% file: postamble-eptcs.inc.tex
\bibliographystyle{abbrv}
\bibliography{paper}
\input appendices.tex

%% file: appendices.tex


%% file: paper.bbl
\begin{thebibliography}{1}
\providecommand{\bibitemdeclare}[2]{}
\providecommand{\surnamestart}{}
\providecommand{\surnameend}{}
\providecommand{\urlprefix}{Available at }
\providecommand{\url}{\begingroup\catcode`_=12 \dourl}
\providecommand{\dourl}[1]{\texttt{#1}\endgroup}
\providecommand{\href}{\begingroup\catcode`_=12 \dohref}
\providecommand{\dohref}[2]{\texttt{#2}\endgroup}
\providecommand{\urlalt}{\begingroup\catcode`_=12 \dourlalt}
\providecommand{\dourlalt}[2]{\href{#1}{#2}\endgroup}
\providecommand{\doi}{\begingroup\catcode`_=12 \dodoi}
\providecommand{\dodoi}[1]{doi:\urlalt{http://dx.doi.org/#1}{#1}\endgroup}
\providecommand{\bibinfo}[2]{#2}

\bibitemdeclare{techreport}{css-ref}
\bibitem{css-ref}
\bibinfo{author}{Bert \surnamestart Bos\surnameend}, \bibinfo{author}{Tantek
  \surnamestart \c{C}elik\surnameend}, \bibinfo{author}{Ian \surnamestart
  Hickson\surnameend} \& \bibinfo{author}{H\o{a}kon~Wium \surnamestart
  Lie\surnameend} (\bibinfo{year}{2011}): \emph{\bibinfo{title}{Cascading Style
  Sheets Level 2 Revision 1 ({CSS} 2.1) Specification}}.
\newblock \bibinfo{type}{Technical Report}, \bibinfo{institution}{World Wide
  Web Consortium}.
\newblock \bibinfo{note}{\url{https://www.w3.org/TR/CSS2/}}.

\bibitemdeclare{inproceedings}{DBLP:conf/dsom/CouchS03}
\bibitem{DBLP:conf/dsom/CouchS03}
\bibinfo{author}{Alva~L. \surnamestart Couch\surnameend} \&
  \bibinfo{author}{Yizhan \surnamestart Sun\surnameend} (\bibinfo{year}{2003}):
  \emph{\bibinfo{title}{On the Algebraic Structure of Convergence}}.
\newblock In \bibinfo{editor}{Marcus \surnamestart Brunner\surnameend} \&
  \bibinfo{editor}{Alexander \surnamestart Keller\surnameend}, editors: {\sl
  \bibinfo{booktitle}{DSOM}}, {\sl \bibinfo{series}{Lecture Notes in Computer
  Science}} \bibinfo{volume}{2867}, \bibinfo{publisher}{Springer}, pp.
  \bibinfo{pages}{28--40}, \doi{10.1007/978-3-540-39671-0_4}.

\bibitemdeclare{inproceedings}{DBLP:conf/icst/HalleBGB15}
\bibitem{DBLP:conf/icst/HalleBGB15}
\bibinfo{author}{Sylvain \surnamestart Hall{\'{e}}\surnameend},
  \bibinfo{author}{Nicolas \surnamestart Bergeron\surnameend},
  \bibinfo{author}{Francis \surnamestart Guerin\surnameend} \&
  \bibinfo{author}{Gabriel~Le \surnamestart Breton\surnameend}
  (\bibinfo{year}{2015}): \emph{\bibinfo{title}{Testing Web Applications
  Through Layout Constraints}}.
\newblock In: {\sl \bibinfo{booktitle}{ICST}}, \bibinfo{publisher}{{IEEE}}, pp.
  \bibinfo{pages}{1--8}, \doi{10.1109/ICST.2015.7102635}.

\bibitemdeclare{inproceedings}{DBLP:conf/ifip6/HalleWVC06}
\bibitem{DBLP:conf/ifip6/HalleWVC06}
\bibinfo{author}{Sylvain \surnamestart Hall{\'{e}}\surnameend},
  \bibinfo{author}{{\'{E}}ric \surnamestart Wenaas\surnameend},
  \bibinfo{author}{Roger \surnamestart Villemaire\surnameend} \&
  \bibinfo{author}{Omar \surnamestart Cherkaoui\surnameend}
  (\bibinfo{year}{2006}): \emph{\bibinfo{title}{Self-configuration of Network
  Devices with Configuration Logic}}.
\newblock In \bibinfo{editor}{Dominique \surnamestart Ga{\"{\i}}ti\surnameend},
  \bibinfo{editor}{Guy \surnamestart Pujolle\surnameend},
  \bibinfo{editor}{Ehab~S. \surnamestart Al{-}Shaer\surnameend},
  \bibinfo{editor}{Kenneth~L. \surnamestart Calvert\surnameend},
  \bibinfo{editor}{Simon~A. \surnamestart Dobson\surnameend},
  \bibinfo{editor}{Guy \surnamestart Leduc\surnameend} \& \bibinfo{editor}{Olli
  \surnamestart Martikainen\surnameend}, editors: {\sl
  \bibinfo{booktitle}{AN}}, {\sl \bibinfo{series}{Lecture Notes in Computer
  Science}} \bibinfo{volume}{4195}, \bibinfo{publisher}{Springer}, pp.
  \bibinfo{pages}{36--49}, \doi{10.1007/11880905_4}.

\bibitemdeclare{inproceedings}{DBLP:conf/sat/NadelR12}
\bibitem{DBLP:conf/sat/NadelR12}
\bibinfo{author}{Alexander \surnamestart Nadel\surnameend} \&
  \bibinfo{author}{Vadim \surnamestart Ryvchin\surnameend}
  (\bibinfo{year}{2012}): \emph{\bibinfo{title}{Efficient {SAT} Solving under
  Assumptions}}.
\newblock In \bibinfo{editor}{Alessandro \surnamestart Cimatti\surnameend} \&
  \bibinfo{editor}{Roberto \surnamestart Sebastiani\surnameend}, editors: {\sl
  \bibinfo{booktitle}{SAT}}, {\sl \bibinfo{series}{Lecture Notes in Computer
  Science}} \bibinfo{volume}{7317}, \bibinfo{publisher}{Springer}, pp.
  \bibinfo{pages}{242--255}, \doi{10.1007/978-3-642-31612-8_19}.

\bibitemdeclare{article}{DBLP:journals/jnsm/NarainLMK08}
\bibitem{DBLP:journals/jnsm/NarainLMK08}
\bibinfo{author}{Sanjai \surnamestart Narain\surnameend}, \bibinfo{author}{Gary
  \surnamestart Levin\surnameend}, \bibinfo{author}{Sharad \surnamestart
  Malik\surnameend} \& \bibinfo{author}{Vikram \surnamestart Kaul\surnameend}
  (\bibinfo{year}{2008}): \emph{\bibinfo{title}{Declarative Infrastructure
  Configuration Synthesis and Debugging}}.
\newblock {\sl \bibinfo{journal}{J. Network Syst. Manage.}}
  \bibinfo{volume}{16}(\bibinfo{number}{3}), pp. \bibinfo{pages}{235--258},
  \doi{10.1007/s10922-008-9108-y}.

\end{thebibliography}
